\documentclass[11pt]{amsart}
\usepackage[foot]{amsaddr}
\usepackage[utf8x]{inputenc}
\usepackage{color}
\usepackage{amsmath,amsfonts}
\usepackage{amsfonts,amsmath,amsthm,amssymb,latexsym}
\usepackage{tikz}
\usepackage{mathrsfs}
\usepackage{mathtools}

\usepackage{algorithm}
\usepackage{algorithmic}
\usepackage{ulem}

\usepackage{hyperref}

\newtheorem{theorem}{Theorem}[section]
\newtheorem{proposition}[theorem]{Proposition}
\newtheorem{lemma}[theorem]{Lemma}
\newtheorem{corollary}[theorem]{Corollary}
\theoremstyle{definition}
\newtheorem{definition}[theorem]{Definition}
\newtheorem{example}[theorem]{Example}

\theoremstyle{remark}
\newtheorem{remark}[theorem]{Remark}
\numberwithin{equation}{section}
\theoremstyle{remark}

\DeclareMathOperator{\lc}{lc}

\DeclareMathOperator{\ddef}{def}
\DeclareMathOperator{\proper}{proper}
\DeclareMathOperator{\nonZ}{nonZ}
\DeclareMathOperator{\surj}{surj}
\DeclareMathOperator{\sqrfree}{sqrfree}
\newcommand{\C}{\mathbb{C}}
\newcommand{\K}{\mathbb{K}}
\newcommand{\LL}{\mathbb{L}}
\newcommand{\FF}{\mathbb{F}}
\newcommand{\Q}{\mathbb{Q}}

\newcommand{\N}{\mathbb{N}}
\newcommand{\V}{\mathbb{V}}
\newcommand{\Pa}{\mathcal{P}}
\newcommand{\PS}{\mathcal{S}}
\newcommand{\Cu}{\mathcal{C}}

\newcommand{\Res}{\mathrm{res}}

\allowdisplaybreaks

\begin{document}

\date{\today}
\title{Rational Solutions of Parametric First-Order Algebraic Differential Equations}

\author{Sebastian Falkensteiner}
\address{Max Planck Institute for Mathematics in the Sciences, Leipzig, Germany.}
\email{sebastian.falkensteiner@mis.mpg.de}

\author{J.Rafael Sendra}
\address{CUNEF Universidad, Departamento de Matemáticas, Spain}
\email{jrafael.sendra@cunef.edu}

{\color{red} The journal version of this paper appears in \\
\textit{[Sebastian Falkensteiner, J. Rafael Sendra,
Rational solutions of parametric first-order algebraic differential equations,
Bulletin des Sciences Mathématiques,
Volume 205,
2025,
103694,
ISSN 0007-4497,
https://doi.org/10.1016/j.bulsci.2025.103694. \\
(https://www.sciencedirect.com/science/article/pii/S0007449725001204)]}
\\
under the CC BY license ( \url{http://creativecommons.org/licenses/by/4.0/} ).}

\begin{abstract}
In this paper, we give an algorithm for finding general rational solutions of a given first-order ODE with parametric coefficients that occur rationally. 
We present an analysis, complete modulo Hilbert's irreducibility problem, of the existence of rational solutions of the differential equation, with parametric coefficients, when the parameters are specialized.
\end{abstract}

\maketitle


\section{Introduction}
Let $\K$ be a computable field of characteristic zero and let $a_1,\ldots,a_n$ be unspecified parameters that eventually will take values in the algebraic closure $\overline{\K}$ of $\K$. 
Let $$\LL:= \K(a_1,\ldots,a_n).$$  
In this work we study first-order algebraic differential equations (AODE) of the form
\begin{equation}\label{eq-F}
F(y,y') = 0 ~\text{ with }~ F \in {\LL}[y,y']
\end{equation}
where $F$ is assumed to be irreducible over the algebraic closure {$\overline{\LL}$ of $\LL$}. 

For this purpose, we will compute symbolically with $a_i$ and analyze the behavior, in terms of rational solutions, of~\eqref{eq-F} when the parameters $a_i$ are specialized. For properly defining the evaluation, we use a parameter space $\PS$ such that
\begin{equation}\label{parameterSpace}
\PS:=\overline{\K}^{\,n}.
\end{equation}
In addition,  we might assume that $\PS$ is an algebraic subset of $\overline{\K}^{\,n}$ implicitly defined by polynomial relations among the $a_1,\ldots,a_n$.

First-order AODEs have been studied extensively, and there are several solution methods for special classes of them. 
However, most of them do not work with differential equations involving coefficients that depend on unknown parameters.

In the case of polynomial coefficients in $x$, Eremenko~\cite{eremenko1998rational} provides a degree bound for rational solutions and hence a method for determining them. 
A more efficient method has been introduced in~\cite{feng2004rational,feng2006polynomial} for autonomous first-order AODEs by associating an algebraic set to the given AODE. 
Then the well-known theory on algebraic curves can be used for finding properties of the rational solutions which help to simplify the differential problem and actually find the solutions. 
The extension to rational general solutions of first-order non-autonomous AODEs can be found in~\cite{RISC4106,RISC5400,RISC5589}. 
For algebraic solutions, we refer to~\cite{aroca2005algebraic,vo2015algebraic}. 
Local solutions of first-order autonomous AODEs are treated in~\cite{cano2019existence}.
For a wide panoramic vision of this algebraic geometry approach, we refer to the survey paper~\cite{falkensteiner2023algebro}.

We follow the algebraic-geometric approach. 
We consider the algebraic curve implicitly defined by the given first-order differential equation by viewing $y$ and $y'$ as independent variables. 
Algebraic curves involving parameters are treated in~\cite{falkensteiner2023rationality} and the results therein play here a crucial role as theoretical and algorithmic tools. Then, by considering the differential relation again, an associated differential equation can be derived. 

In this paper, we present an almost complete analysis of the existence of rational solutions of the differential equations of \eqref{eq-F} when the parameters are specialized in $\PS$. 
The main difficulty to provide a complete analysis is that one needs to deal with the (open) irreducibility problem of Hilbert. 
Note that a preliminary condition for most methods to find non--trivial rational solutions is that the associated curve should be rational and, hence, irreducible. Nevertheless, we provide an isolation of such specializations. 
More precisely, we provide a decomposition of the parameter space $\PS$ into three subsets:
\begin{itemize}
	\item The first subset contains the cases where the specialization of the differential polynomial is either not well--defined or it is a constant.
	\item The second subset contains the cases where the specialized differential equation either degenerates to a positive genus curve or it is reducible or, although generating a genus zero curve, it does not have non-trivial rational solutions.
	\item The third subset provides the specializations where the new different equation has non-trivial rational solutions.
\end{itemize}

The structure of the paper is as follows.
In Section~\ref{sec-parametricCurves} we fix notation and summarize the relevant results on parametric rational curves. 
In particular, the decomposition of the parameter space under the criterium of proving rational parametrizations is treated (see Subsection~\ref{sec-decomposition} and also Appendix ~\ref{sec-decomposition-surjectivity}). 

In Section~\ref{sec-autonomousParameters} we generalize previous results on first-order autonomous differential equations to the parametric case. 
This behavior is depending on the exact values of the constant parameters. 
We give a finite decomposition of the parameter space where the solvability is unchanged and general rational solutions can be computed whenever they exist (Theorems~\ref{thm:FengGaospec},\ref{thm:solCovering},\ref{thm:alg1}). 
We illustrate the algorithmic method by examples. In addition, we include an appendix where the decomposition of the parameter space w.r.t. rational parametrizations is extended to the case of providing a rational covering (Theorem~\ref{thm-surjectiveSpecialization}).

\section{Parametric rational curves}\label{sec-parametricCurves}
Let us first fix notations and recall some results on rational curves; for further details see~\cite{Cox1,SWP08}. In the two remaining subsections we analyze the behavior, under specializations, of parametric rational curves.

\subsection{Preliminaries}\label{sec-pre}
For a field $\mathcal{K}$, we denote by $\overline{\mathcal{K}}$ its algebraic closure.
We will express tuples with bold face letters. For instance, the tuple of undetermined parameters will be expressed as $\textbf{a}=(a_1,\ldots,a_n)$. Furthermore,  we will usually set $\LL=\K(\textbf{a})$ and $\FF=\LL(\delta)$ where $\delta$ is an algebraic element over $\LL$.

Let $F(y,y') \in \LL[y,y']$ be an irreducible polynomial (over $\overline{\LL}$) and depending on $y'$.
Then we define the \textit{associated curve to $F$} as the zero-set of $F$ over $\overline{\LL}$, i.e. 
\[ \Cu(F) = \{ (p,q) \in \overline{\LL}^2 \mid F(p,q)=0 \} .\]
A \textit{(rational) parametrization} of $\Cu(F)$ is a pair $\Pa(t) \in \overline{\LL}(t)^2 \setminus \overline{\LL}^2$ such that $F(\Pa(t))=0$ holds.
A rational parametrization of $\Cu(F)$ exists if and only if the genus of the curve is equal to zero~\cite[Theorem 4.63]{SWP08}. 
If $\Pa$ is birational, then $\Pa$ is called a \textit{proper} or \textit{birational} parametrization.
If $\Cu(F)$ admits a parametrization, we say that $\Cu(F)$ is a rational curve.

Let us note that for the differential part of the problem it would not be necessary to require that $F$ is irreducible. 
But only in this case $\Cu(F)$ can admit a rational parametrization~\cite[Theorem 4.4]{SWP08}.

In general, if one computes a parametrization $\Pa(t)$ of $\Cu(F)$, the ground field $\LL$ has to be extended. 
The coefficient field of $\Pa(t)$ is called the \textit{field of definition}. 
Moreover, a subfield $\FF$ of $\overline{\LL}$ is called a \textit{field of parametrization} of $\Cu(F)$ if there exists a parametrization with $\FF$ as field of definition.

One can achieve a field of parametrization $\FF=\LL(\delta)$, for some $\delta^2 \in \LL$, of $\Cu(F)$ as it is highlighted in the following theorem (see ~\cite{hilbert1890diophantischen} and Theorem 5.8. and Corollary 5.9. in \cite{SWP08}).
\begin{theorem}\label{theorem:HHext} \
\begin{enumerate}
	\item If $\deg(\Cu(F))$ is odd then $\LL$ is a field of parametrization.
	\item If $\deg(\Cu(F))$ is even then either $\LL$ is a field of parametrization or there exists $\delta\in \overline{\LL}$ algebraic over $\LL$, with minimal polynomial $t^2-\alpha\in \LL[t]$, such that $\LL(\delta)$ is a field of parametrization of $\Cu(G)$.
\end{enumerate}
\end{theorem}

The notation for the evaluation of parameters $\textbf{a}$ will simply consist in replacing the parameters by $\textbf{a}^0\in \PS$ and, if the dependencies on $\textbf{a}$ are not explicitly stated, by prepending $\textbf{a}^0$ in the argument. 
At some steps, it might be necessary to work with the field extension provided by an algebraic element $\gamma(\textbf{a}) \in \overline{\K(a_1,\ldots,a_n)} \setminus \overline{\K}(a_1,\ldots,a_n)$ that depends on the parameters $\textbf{a}$.
Also in this case, we will simply write the dependencies on $\textbf{a}$ and not explicitly state $\gamma(\textbf{a})$ in the argument. For given  $f,g \in \overline{\K(\textbf{a})}[\textbf{z}]$, we denote by $\Res_{z_0}(f,g)$ the resultant of $f$ and $g$ with respect to the variable $z_0$ among $\textbf{z}$.

In addition, throughout the paper we use some Zariski-open subsets of the parameter space $\PS$ (see~\eqref{parameterSpace}), to be considered when specializing the parameters, that have been defined in~\cite{falkensteiner2023rationality}. More precisely, for  $R, f, g \in \overline{\K}(\textbf{a})[\textbf{z}]$ we use:
\begin{enumerate}
	\item  The set $\Omega_{\ddef(R)}$ where $R$ is defined under specialization, and $\Omega_{\nonZ(R)}$ where $R$ is defined and non-zero under specialization. If $R$ is the defining polynomial of an algebraic curve, we additionally ask the degree of $R$ w.r.t. $\textbf{z}$ to be preserved under specialization
    (see~\cite[Definition 3.3]{falkensteiner2023rationality}).
	\item The set $\Omega_{\gcd(f,g)}$ (see~\cite[Definition 3.6]{falkensteiner2023rationality}). 
	When specializing in this open subset, the gcd behaves properly; that is, the gcd of the specialized polynomials is the specialization of the gcd, and the degree of the gcd does not change after evaluation.
	\item For squarefree $f$, we use $\Omega_{\sqrfree(f)}$ as in~\cite[Definition 3.12]{falkensteiner2023rationality} such that for $\textbf{a}^0 \in \Omega_{\sqrfree(f)}$ it holds that $f(\textbf{a}^0;\textbf{z})$ is square-free. 
	\item For a given proper rational parametrization in reduced form $\Pa=(p_1/q_1,p_2/q_2) \in \overline{\mathbb{L}}(t)^2$ of $\Cu(F)$, we define $${\Omega_{\ddef(\Pa)}=\Omega_{\ddef(p_1)}\cap \Omega_{\ddef(p_2)}\cap \Omega_{\nonZ(q_1)}\cap   \Omega_{\nonZ(q_2)}}.$$ Moreover, we consider the open set $\Omega_{\proper(\Pa)} \subseteq \PS$ as in~\cite[Definition 5.1 and Definition 5.4]{falkensteiner2023rationality} such that every specialization $\textbf{a}^0 \in \Omega_{\proper(\Pa)}$ satisfies that $\Pa(\textbf{a}^0;t)$ is a proper parametrization of the specialized curve $\Cu(\textbf{a}^0;F)$.  We observe that $\Omega_{\ddef(\Pa)}\subseteq \Omega_{\proper(\Pa)}$ and $\Omega_{\ddef(F)}\subseteq \Omega_{\proper(\Pa)}$. 
\end{enumerate}
 For a prime ideal $\mathcal{I}$ of $\K[\textbf{a}]$, we will denote the quotient field of the unique factorization domain $\overline{\K}[\textbf{a}]/\mathcal{I}$ as $\mathcal{K}:=\mathcal{Q}(\overline{\K}[\textbf{a}]/\mathcal{I})$. 
Then the coefficients of the polynomials $R, f, g$ might also be considered as the canonical representatives in $\mathcal{K}$ and specializations are taken in $\textbf{a}^0 \in \V(\mathcal{I}) \subset \PS$ (see also~\cite[Section 6]{falkensteiner2023rationality}).

\subsection{Decomposition w.r.t. parametrizations}\label{sec-decomposition}
Let us now give further details on rational parametrizations of families of rational curves depending on several unspecified parameters. In particular, the idea is to decompose the space $\PS$ (see~\eqref{parameterSpace}), where the parameters take values, in constructible subsets where the curve, when the parameters are specialized, satisfies certain properties. For details on parametric rational curves we refer to~\cite{falkensteiner2023rationality}.

In the sequel, $\Cu(F)$ is an irreducible algebraic curve defined by $F\in \LL[y,z]$. We start with the following definitions.

\begin{definition}\label{def-dec-param-prev}
Let $\PS^*\subset \PS$ and let $\Pa(t) \in \overline{\LL}(t)^2$ be a proper rational parametrization of $\Cu(F)$. We say that $\Pa(t)$ is \textit{$\PS^*$--admissible} if for all $\textbf{a}^0\in \PS^{*}$ it holds that $\Pa(\textbf{a}^0;t)$ is a proper rational parametrization of $\Cu(\textbf{a}^0; F)$. In particular, $\Pa(\textbf{a}^0;t)$ is defined and not constant.
\end{definition}


\begin{definition}\label{def-dec-param} Let  $F\in \LL[y,z]$ be irreducible. 
Let $I\subset \N$ be finite. For $i\in I$, let $\PS_1,\PS_2,\PS_{3,i}\subset \PS$ be disjoint  constructible sets, let 
\[ \PS_{3}=\displaystyle{\dot{\bigcup}_{i\in I} \PS_{3,i}}, \]
and let $\Pa_i(\textbf{a}^0;t)$ be $\PS_{3,i}$--admissible parametrizations. 
We say that
\begin{equation}\label{eq-decomp2}
\PS_1 \,\dot\cup\, \PS_2\, \dot\cup\, \PS_3,
\end{equation} 
is a \textit{decomposition of $\PS$ w.r.t. (rational) parametrizations of $\Cu(F)$} if  
\begin{itemize}
\item[(a)] $\displaystyle{\PS=\PS_1 \,\dot\cup\, \PS_2\,\dot\cup\, \PS_3.}$
\item[(b)]
For every specialization  $\textbf{a}^0\in \PS_j$, case (j) below holds:
\begin{enumerate}
    \item either $F(\textbf{a}^0;y,z)$ is not well--defined or $F(\textbf{a}^0;y,z)\in \overline{\K}$; in this case, we say that the specialization degenerates;
    \item the genus of $\Cu(\textbf{a}^0;F)$ is positive, or $F(\textbf{a}^0;y,z)$ is reducible (over $\overline{\K}$);
    \item the genus of $\Cu(\textbf{a}^0;F)$ is zero and $\Pa_i(\textbf{a}^0;t)$ is a proper parametrization of $\Cu(\textbf{a}^0;F)$.\end{enumerate}
\end{itemize}
\end{definition}

\begin{remark}\label{rem-def-dec-param} \,
\begin{enumerate}
\item With abuse of notation, we might write $\{(\PS_{3,i},\Pa_i)\}$ instead of $\PS_{3,i}$.
\item The decomposition of $\PS_3$ is, in general, not unique since it depends in particular on the chosen rational parametrizations.
\item As a result of the process described in~\cite{falkensteiner2023rationality} (in particular, Section 6), a decomposition of $\PS$ w.r.t. (rational) parametrizations always exists and can be computed algorithmically.
\end{enumerate}
\end{remark}

\begin{proposition}\label{cor:componentpara}
Let $F \in \LL[y,z]$ be irreducible and let $\PS=\PS_1 \dot\cup \PS_2 \dot\cup \PS_3$, $\PS_{3}=\dot{\bigcup}_{i\in I} \PS_{3,i}$ be a decomposition w.r.t. parametrizations of $\Cu(F)$. 
Let $\textbf{a}^0 \in \PS$ be such that $\Cu(\textbf{a}^0;F)$ is rational. Let $\tilde{\Pa} \in \overline{\K}(t)^2$ be a rational parametrization of $\Cu(\textbf{a}^0;F)$. 
Then, there exists a component $\PS_{3,i}$ with a corresponding parametrization $\Pa_i(t)$ such that $\tilde{\Pa}$ is a reparametrization of $\Pa_i$, i.e. $\Pa_i(\textbf{a}^0;s) = \tilde{\Pa}(t)$ for some non--constant $s \in \overline{\K}(t)$.
\end{proposition}
\begin{proof}
By assumption, $\Cu(\textbf{a}^0;F)$ is rational. 
Thus, by construction, there exists $(\PS_{3,i},\Pa_i)$, with $\textbf{a}^0 \in \PS_{3,i}$ such that $\Pa_i(\textbf{a}^0;t)$ is a proper parametrization of $\Cu(\textbf{a}^0;F)$. 
By~\cite[Lemma 4.17]{SWP08}, $\tilde{\Pa}$ is a reparametrization of $\Pa_i(\textbf{a}^0;t)$.
\end{proof}

Finally, let us note that in the case of a single parameter, the field of definition can be chosen as the base field:
\begin{remark}\label{rem:oneparameterfieldofparametrization}
Based on Tsen's theorem~\cite{ding1999chiungtze}, in the case of $n=1$ and $F \in \K(a)[y,z]$ defining the rational curve $\Cu(F)$, $\K(a)$ is a field of parametrization~\cite[Corollary 2.3]{falkensteiner2023rationality}.
\end{remark}

\section{Differential equations with constant parameters}\label{sec-autonomousParameters}
{ In this section, we will use some of the notation introduced previously:
We write in bold letters tuples such as for the unspecified parameters $\textbf{a}=(a_1,\ldots,a_n)$, the parameter space where $\textbf{a}$ are evaluated as $\PS$, and the evaluation of an expression $g(x,y)$ with coefficients in $\textbf{a}$ as $g(\textbf{a}^0;x,y)=g(x,y)|_{\textbf{a}=\textbf{a}^0}$. 
Moreover, $\K$ denotes a computable field of characteristic zero, $\LL=\K(\textbf{a})$ is the field of rational functions in the variables $\textbf{a}$, and for a given irreducible $F \in \LL[y,y']$ we denote by $\Cu(F)$ the corresponding algebraic curve over $\overline{\LL}$. 
The field $\LL$ may also be chosen as $\FF_\mathrm{J}$ for some prime ideal $\mathrm{J}$ as explained below in Remark~\ref{rem-eq-decomp3}(3).

The following two statements, Lemma~\ref{lem-nec} and Theorem~\ref{thm:FengGao}, follow by the same proof as of Theorems 2 and 5 in~\cite{feng2004rational}, respectively, by replacing the coefficient field $\Q$ with $\LL$.

\begin{lemma}\label{lem-nec}
Let $y(x) \in \overline{\LL}(x)$ be a solution of $F(y,y')=0$ where $F \in \LL[y,y']$ is irreducible. 
Then $(y(t),y'(t))$ is a proper rational parametrization of $\Cu(F)$.
\end{lemma}

Since all proper rational parametrizations of $\Cu(F)$ are related by a Möbius-transformation, after a careful analyzation including the derivative, the following can be shown.

\begin{theorem}\label{thm:FengGao}
Let $F \in \LL[y,y']$ be irreducible and let $\Pa(t)=(P_1(t),P_2(t)) \in \LL(t)^2$ be a proper parametrization of $\Cu(F)$. 
Then there is a  non-constant rational solution of $F(y,y')=0$ if and only if either
\begin{equation}\label{eq:FengGaoCases}
\alpha\,P_1'(t)=P_2(t) \ \text{ or } \ \alpha\,(t-\beta)^2\,P_1'(t)=P_2(t) 
\end{equation}
for some $\alpha, \beta \in \LL$ with $\alpha \ne 0$. 
In the affirmative case, $y(x)=P_1(\alpha \cdot (x+c))$ (or $y(x)=P_1(\beta-\frac{1}{\alpha \cdot (x+c)})$), where $c \in \overline{\LL}$ is an arbitrary constant, defines all rational solutions of $F(y,y')=0$.
\end{theorem}

The previous theorem motivates the next definition.

\begin{definition}\label{def-sol}
With the notation of Theorem~\ref{thm:FengGao}, if $\Pa(t)$ satisfies~\eqref{eq:FengGaoCases}, we call $$y(x)=P_1(\alpha \cdot (x+c)),\,\,\text{or}\,\, y(x)=P_1\left(\beta-\frac{1}{\alpha \cdot (x+c)}\right),$$ depending on the case, the \textit{rational  {(general)} solution generated by $\Pa(t)$}.
\end{definition}

\begin{remark}\label{rem-zero}
Let us note that if $\alpha=0$ in the first case of Theorem~\ref{thm:FengGao}, we obtain a constant solution given as $y(x)=P_1(0)$. 
Not all constant solutions of $F(y,y')=0$, however, might be found in this way.
\end{remark}

The following is an adapted version of~\cite[Theorem 6]{feng2004rational} to our setting justifying to consider solutions without field extensions involving the parameters.

\begin{theorem}\label{thm:solutionExtension}
Let $F \in \LL[y,y']$ be irreducible. 
If there exists a  {non-constant} solution $y(x) \in \overline{\LL}(x)$ of $F(y,y')=0$, then there is another solution $z(x) \in \LL(x)$.
\end{theorem}

\subsection{Specializations of the parameters}\label{subsec-FG}
Let $F(y,y')$ be a differential polynomial as in~\eqref{eq-F}. 
We now study rational solutions of $F({\textbf{a}^0};y,y')$ for $\textbf{a}^0\in \PS$. 

\begin{definition}
   Let $F \in \K(\textbf{a})[y,y']$ be as in~\eqref{eq-F}, $y(x) \in \overline{\LL}(x)$ be a non-constant rational solution of $F(y,y')=0$. Let $\mathcal{S}^*\subset \mathcal{S}$. We say that $y(x)$ is $\mathcal{S}^{*}$--admissible if for every $\textbf{a}^0\in \mathcal{S}^*$ it holds that 
   $F(\textbf{a}^0;y,z)$ and $y(\textbf{a}^0;x)$ are both well defined, and $y(\textbf{a}^0;x)$ is a non-constant rational solution of $F(\textbf{a}^0;y,y')=0$. 
\end{definition}

\begin{proposition}\label{prop:solSpec}
Let $F \in \LL[y,y']$ be as in~\eqref{eq-F}, $y(x) \in  {\overline{\LL}}(x)$ be a non-constant rational solution of $F(y,y')=0$, $\Pa=(y(t),y'(t))$, and let ${\textbf{a}^0} \in \Omega_{\ddef(F)} \cap \Omega_{\ddef(\Pa)}$ (see Subsec.~\ref{sec-pre}).
Then $y({\textbf{a}^0};x)$ is a well--defined rational solution of $F({\textbf{a}^0};y,y')=0$. Moreover, if $\textbf{a}^0 \in \Omega_{\proper(\Pa)}$, then $y({\textbf{a}^0};x)$ is a non-constant rational solution of $F( {\textbf{a}^0};y,y')=0$, that is, $y(x)$ is $ \Omega_{\proper(\Pa)}$--admissible  (see Subsec.~\ref{sec-pre} for the notion of $\Omega_{\proper(\Pa)}$).
\end{proposition}
\begin{proof} Let ${\textbf{a}^0} \in \Omega_{\ddef(F)} \cap \Omega_{\ddef(\Pa)}$. 
By assumption, the specialization of $\Pa( {\textbf{a}^0};t)$ remains a zero of $F( {\textbf{a}^0};y,y')$. So if $F( {\textbf{a}^0};y,y')$ is a constant, then it is identically zero and the statement trivially holds. 
Let the specialization $F( {\textbf{a}^0};y,y')$ be non-degenerate. 
Since the second component of $\Pa( {\textbf{a}^0};t)$ remains the derivative of the first and $\Pa( {\textbf{a}^0};t)=(y( {\textbf{a}^0};t),y'( {\textbf{a}^0};t))$ is well--defined, $F( {\textbf{a}^0};y( {\textbf{a}^0};x),y'( {\textbf{a}^0};x))=0$. Now, let $\textbf{a}^0 \in  \Omega_{\proper(\Pa)}$. Then, $\Pa( {\textbf{a}^0};t)=(y( {\textbf{a}^0};x),y'( {\textbf{a}^0};x))$ is a (proper) parametrization, and hence $y( {\textbf{a}^0};x)$ cannot be a constant. Thus, $y(x)$ is $ \Omega_{\proper(\Pa)}$--admissible. 
\end{proof}

Proposition~\ref{prop:solSpec} treats the case where a rational solution of $F$ exists. 
In the following, we analyze the cases where $y({\textbf{a}^0};x)$ is not well--defined or $F(y,y')$ itself does not admit a rational solution. 
We show how all solutions under those specializations can be found where $F({\textbf{a}^0};y,z)$ remains irreducible and represent the solution set in a finite way. 
Observe that the problem of algorithmically finding the parameters ${\textbf{a}^0} \in \PS$ such that $F({\textbf{a}^0};y,z)$ is reducible is an open problem (see e.g.~\cite{Hilbert1892} and \cite{Serre1997}), but the decomposition~\eqref{eq-decomp2} provides an isolation of such specializations (they are in $\PS_2$).

Similarly as in~\cite{falkensteiner2023rationality}, let us decompose the parameter space such that the behavior for every specialization in a component is the same. 

\begin{theorem}\label{thm:FengGaospec}
Let $F \in \LL[y,y']$ be as in~\eqref{eq-F} and let $\Pa=(P_1,P_2) \in \overline{\LL}(t)^2$ be a proper parametrization of $\Cu(F)$. 
Then the following holds.
\begin{enumerate}
    \item If $\Pa$ fulfills~\eqref{eq:FengGaoCases} for some $\alpha, \beta \in \overline{\LL}$ with $\alpha \ne 0$ leading to the rational solution $y(x)$ of $F(y,y')=0$, then for every $\textbf{a}^0 \in  \Omega_{\proper(\Pa)}$ (see Subsec.~\ref{sec-pre}) it holds that $y(\textbf{a}^0;x)$ is a rational solution of $F(\textbf{a}^0;y,y')=0$.
    \item If $\Pa$ does not fulfill~\eqref{eq:FengGaoCases}, let $A/B=P_2/P_1'$ be such that $A,B$ are coprime.  Then for every $\textbf{a}^0\in\hat{\Omega}$, where
    $$\hat{\Omega}:=  \Omega_{\nonZ(\Res_t(A,B))} \cap \Omega_{\nonZ(\lc(A))} \cap \Omega_{\nonZ(\lc(B))} \cap \Omega_{\sqrfree(A/B)}  \cap \Omega_{\mathrm{proper}(\Pa)},$$ also $\Pa(\textbf{a}^0;t)$ does not fulfill~\eqref{eq:FengGaoCases}.
\end{enumerate}
\end{theorem}
\begin{proof}
Item (1) holds due to Proposition~\ref{prop:solSpec}. 
For item (2), note that for every $\textbf{a}^0 \in \Omega_{\mathrm{FG}}$ it holds that $A,B$ have the same degrees as $A(\textbf{a}^0), B(\textbf{b})$, respectively, and $A(\textbf{a}^0), B(\textbf{a}^0)$ are again coprime. 
Thus, if $A/B$ is not a polynomial of degree zero or two, then this is neither the case for $A(\textbf{a}^0)/B(\textbf{a}^0)$. 
So it remains to consider those two cases. 
If $A(\textbf{a}^0)/B(\textbf{a}^0)$ is constant, then this was already the case for $A/B$, in contradiction to the assumption that $A/B=P_2/P_1'$ does not fulfill the first condition in~\eqref{eq:FengGaoCases}. 
If $A/B$ (and $A(\textbf{a}^0)/B(\textbf{a}^0)$) is a polynomial of degree two, then it has to be square-free; otherwise it would fulfill the second condition in~\eqref{eq:FengGaoCases}. 
Since in this case we assume that $\textbf{b} \in \Omega_{\sqrfree(A,B)}$, also $A(\textbf{a}^0)/B(\textbf{a}^0)$ is square-free and hence, does not fulfill~\eqref{eq:FengGaoCases}.
\end{proof}

\begin{remark}\label{rem-onFengGaoSpec}
Note that the contraposition of item (2) in Theorem~\ref{thm:FengGaospec} is: If there exists $\textbf{a}^0 \in  \hat{\Omega}$ such that $\Pa(\textbf{a}^0;t)$ fulfills~\eqref{eq:FengGaoCases}, then $\Pa$ fulfills~\eqref{eq:FengGaoCases}. Moreover, the case fulfilled in~\eqref{eq:FengGaoCases} remains the same for specializations and generalizations of the parameters.
\end{remark}

 Based on this result we introduce the following definition.

\begin{definition}\label{def-FG}
With the notation of Theorem~\ref{thm:FengGaospec}, we define 
\begin{equation*}
\Omega_{\mathrm{FG}}=
\begin{cases}
\Omega_{\proper(\Pa)} & \text{if $\Pa$ fulfills~\eqref{eq:FengGaoCases},}
\\
\hat{\Omega} & \text{if $\Pa$ does not fulfill~\eqref{eq:FengGaoCases}.}
\end{cases}
\end{equation*}
\end{definition}}
 Also, we define the following decomposition of $\PS$.

\begin{definition}\label{def-dec-ratsol}
Let $F(y,y')=0$ be as in~\eqref{eq-F}.  
Let $I\subset \N$ be finite. For $i\in I$, let ${\PS}_{1}^{*}, {\PS}_{2}^{*},{\PS}_{3,i}^{*}\subset \PS$ be disjoint constructible sets, let 
\[ {\PS}_{3}^{*}=\displaystyle{\dot{\bigcup}_{i\in I} {\PS}_{3,i}^{*}}, \] 
and let $y_i(\textbf{a},x)\in \overline{\LL}(x)$ be $\PS_{3,i}^{*}$--admissible solutions. 
We say that 
\begin{equation}\label{eq-decomp4}
 {\PS}_{1}^{*} \,\dot\cup\, {\PS}_{2}^{*}\, \dot\cup\, {\PS}_{3}^{*}, 
\end{equation} 
is a \textit{decomposition w.r.t. rational solutions of $F(y,y')=0$} if  
\begin{itemize}
\item[(a)] $\displaystyle{\PS= {\PS}_{1}^{*} \,\dot\cup\,{ \PS}_{2}^{*}\,\dot\cup\, {\PS}_{3}^{*}.}$
\item[(b)]
For every specialization $\textbf{a}^0\in {\PS}_{j}^{*}$, case (j) holds:
\begin{enumerate}
    \item either $F(\textbf{a}^0;y,y')$ is not well--defined or $F(\textbf{a}^0;y,y')\in \overline{\K}$;
    \item either $F(\textbf{a}^0;y,y')=0$ does not have a non-constant rational solution, or $F(\textbf{a}^0;y,y')$ is reducible (over $\overline{\K}$);
    \item $y_j(\textbf{a}^{0};x)$ is well defined and is a non-constant rational solution of $F(\textbf{a}^0;y,y')=0$.
\end{enumerate}
\end{itemize}
\end{definition}

\begin{remark} 
\begin{enumerate}\
\item 
The main structural difference between a decomposition w.r.t. para\-me\-trizations (see Definition~\ref{def-dec-param}) $\PS_1 \cup \PS_2 \cup \PS_3$ and a decomposition w.r.t. rational solutions $\PS_1^* \cup \PS_2^* \cup \PS_3^*$ is that the associated elements to $\PS_3$ are rational parametrizations of the defining algebraic equation and that to $\PS_3^*$ are non-constant rational solutions of the defining differential equation, respectively.
\item If $y_i(\textbf{a}^{0};x)$ is a non-constant rational solution of $F(\textbf{a}^0;y,y')=0$, by Theorem~\ref{thm:FengGao}, 
$y_i(\textbf{a}^{0};x+c)$ is a general rational solution of $F(\textbf{a}^0;y,y')=0$.
\end{enumerate}
\end{remark}

\subsection{Algorithmic treatment}

Using the results in Subsection~\ref{subsec-FG}, 
combined with Section~\ref{sec-parametricCurves}, we can derive an algorithm for computing a decomposition w.r.t. rational solutions.
For this purpose, we will need the computation of a decomposition w.r.t. parametrizations as in Definition~\ref{def-dec-param}, see Remark~\ref{rem-def-dec-param}, and we call this auxiliary algorithm \textsc{ParamDecomposition}. Let us say that this decomposition is 
\begin{equation*}
\PS_1 \,\dot\cup\, \PS_2\, \dot\cup\, \PS_3,
\end{equation*} 
with
\[ \PS_{3}=\displaystyle{\dot{\bigcup}_{i\in I} \PS_{3,i}}. \]
For checking~\eqref{eq:FengGaoCases} in the constructible sets $\PS_{3,i}$, different approaches can be considered. We note that, by construction, $\PS_{3,i}$ could be written as a set of equations and inequations implicitly given by polynomials in $\K[\textbf{a}]$. 
Then one can use a Gr\"obner basis of the defining polynomials, adding the inequations via the Rabinowitsch trick to the system, and reduce the question to an ideal membership problem~\cite[Section 4]{Cox1}. Alternatively, one might compute and work with triangular set decompositions such as (algebraic) Thomas decomposition~\cite{bachler2012algorithmic} or regular chains~\cite{kalk1993}. 
Corresponding to the constructible sets $\PS_{3,i}$ and the parametrizations $\Pa$, we instead choose to work with the open sets
$\Omega_{\mathrm{FG}}$, that are recursively defined over different specialization spaces (see $\Sigma$ in Algorithm~\ref{alg-ConstParameters}). Depending on the behavior of $\Pa$ w.r.t.~\eqref{eq:FengGaoCases}, $\Omega_{\mathrm{FG}}$ is defined differently. If $\Pa$ fulfills~\eqref{eq:FengGaoCases}, then $\Omega_{\mathrm{FG}}=\Omega_{\mathrm{proper}(\Pa)}\supseteq \PS_{3,i}$. However, if $\Pa$ does not fulfill~\eqref{eq:FengGaoCases} then in general $\Omega_{\mathrm{FG}} \not\supseteq \PS_{3,i}$. 
{In this situation, if the parameter subspaces in the output should remain disjoint, additionally $\PS_{3,i} \setminus \Omega_{\mathrm{FG}}$ has to be considered. We work with a closed superset of $\PS_{3,i} \setminus \Omega_{\mathrm{FG}}$ in order to apply a prime decomposition comparable to~\cite[Section 6]{falkensteiner2023rationality} to do so. Otherwise one might get solutions that are covered several times (cf. Example~\ref{ex-1}).} More precisely, we consider the following algorithm. 

\begin{algorithm}[H]
\caption{ConstantParameterSolve}
\label{alg-ConstParameters}
\begin{algorithmic}[1]
    \REQUIRE A first-order AODE $F(y,y')=0$ as in~\eqref{eq-F}.
    \ENSURE A decomposition w.r.t. rational solutions of $F(y,y')=0$ and their associated solutions.
    \STATE By Algorithm~ \textsc{ParamDecomposition}, compute a decomposition w.r.t. parametrizations $$\PS = \dot\bigcup_{j=1}^3\PS_j, \PS_{3} = \dot\bigcup_{i \in I}\PS_{3,i}.$$
    \STATE Set $\PS_1^*=\PS_1, \PS_2^*=\PS_2$, $\PS_{3}^*=\emptyset$, $\mathcal{R}=\overline{\K}[\textbf{a}]$ and $\Sigma=\PS$.
    \STATE For every $i \in I$, perform the following steps.
        \STATE For $\Pa = (P_1(t),P_2(t))$ corresponding to $\PS_{3,i}$, check whether equation~\eqref{eq:FengGaoCases} is fulfilled in $\Sigma$ over $\mathcal{Q}(\mathcal{R})$ with $\alpha \ne 0$.
        \STATE In the positive case, add $\Omega_{\mathrm{FG}} \setminus (\PS_1^* \cup \PS_2^* \cup \PS_3^*)$ and the associated (general) rational solution $y_i(x+c)=P_1(\alpha\,(x+c))$ (or $y_i(x+c)=P_1(\beta-\frac{1}{\alpha\,(x+c)})$) to $\PS_3^*$, and move to the next component.
        \STATE In the negative case add $\Omega_{\mathrm{FG}} \setminus (\PS_1^* \cup \PS_2^* \cup \PS_3^*)$ to $\PS_2^*$.
        \STATE  If $\Sigma \setminus \Omega_{\mathrm{FG}} \ne \emptyset$, decompose the polynomial ideal given by $\Sigma \setminus \Omega_{\mathrm{FG}}$ in its prime components $\mathcal{I}_1,\ldots,\mathcal{I}_{k}$. Otherwise move to the next component.
        \STATE For every $\ell \in \{1,\ldots,k\}$, repeat the loop (4)-(8) with {$\mathcal{R} \leftarrow \mathcal{R}/\mathcal{I}_\ell$ as new base field in step (4)} and $\Sigma=\V(\mathcal{I}_\ell)$ as new parameter space.
    \STATE Return $\PS_{1}=\PS_{1}^*, \PS_2^*, \PS_{3}^*$ and $\{(\PS_{3,i}^*, y_i(x+c))\}_{i \in I^*}$.
\end{algorithmic}
\end{algorithm}

{Let us emphasize that the field $\mathcal{Q}(\mathcal{R})$ is just used as ground field  for performing the corresponding arithmetic operations in Step (4), and in particular for checking equation~\eqref{eq:FengGaoCases}. Note that the original parameter space $\PS$, and the subsequence replacements, are considered over the original base field $\overline{\mathbb{K}}$ throughout the whole algorithm.}

\begin{theorem}\label{thm:alg1}
Algorithm~\ref{alg-ConstParameters} is correct.
\end{theorem}
\begin{proof}
Let us first show that the output is a decomposition w.r.t. rational solutions. 
The $\PS_j^*$ with $j \in \{1,2,3\}$ and the $\PS_{3,i}^*$ with $i \in I^*$, respectively, are by construction disjoint since the $\PS_j$ and $\PS_{3,i}$ are. Moreover, since just intersections, unions and set-complements are used in the construction, $\PS_j^*$ and $\PS_{3,i}^*$ are constructible.

In the following we use tuples of superscripts for the iteration number and the branch in the loop (4)-(8). In step (5),   $\Omega_{\mathrm{FG}}^{(0)} = \Omega_{\mathrm{proper}(\Pa)}$ covers $\PS_{3,i}$. If step (6) is reached, then the loop (4)-(8) gets repeated with the $\Sigma^{(1,\ell)}$ such that $\bigcup_\ell \Sigma^{(1,\ell)}=\Sigma^{(0)} \setminus \Omega_{\mathrm{FG}}^{(0)}$ or $\Sigma^{(0)} \setminus \Omega_{\mathrm{FG}}^{(0)} = \emptyset$. 
Let use iteratively continue with the negative case (6) and set $\Sigma^{(j,k_j)}$ and $\Omega_{\mathrm{FG}}^{(j,k_j)}$ to the latest non-empty set in the iteration.  Let us see that 
the newly added components cover $\PS_{3,i}$ since it is a subset of
\begin{equation*}
\begin{aligned}
&\PS_{3,i} \setminus \left(\bigcup_{j \ge 0, \ell_j \ge 1}\Omega_{\mathrm{FG}}^{(j,\ell_j)}\right) \subseteq \bigcap_{j \ge 0} \PS \setminus \left(\bigcup_{\ell_j \ge 1}\Omega_{\mathrm{FG}}^{(j,\ell_j)}\right) \\ & = \bigcap_{j \ge 0, \ell_j \ge 1} \Sigma^{(j,\ell_j)} \cap (\Sigma^{(j,k_j)} \setminus \Omega_{\mathrm{FG}}^{(j,k_j)}) = \emptyset.
\end{aligned}
\end{equation*}
Thus, $\PS_1^* \cup \PS_2^* \cup \PS_3^* = \PS_1 \cup \PS_2 \cup \PS_3$ and item (a) in Definition~\ref{def-dec-ratsol} is fulfilled.

 On the other hand, for a specialization $\textbf{a}^0 \in \PS$ it holds that:
\begin{itemize}
    \item[(1)] $\PS_1 = \PS_1^*$ because $F(\textbf{a}^0;y,z)$ is not well-defined in both cases.
    \item[(2)] If $\textbf{a}^0 \in \PS_2$, then $F(\textbf{a}^0;y,z)$ does not admit a rational parametrization or it is reducible. Since a non-constant rational solution defines a rational parametrization (see Lemma~\ref{lem-nec}), $\textbf{a}^0 \in \PS_2^*$.
    \item[(3)] If $\textbf{a}^0 \in \PS_3$, then in steps (4)-(8) is decided whether $\textbf{a}^0 \in \PS_3^*$ or $\textbf{a}^0 \in \PS_2^*$. By Theorem~\ref{thm:FengGaospec}, $\textbf{a}^0 \in \PS_3^*$ if and only if $F(\textbf{a}^0;y,y')=0$ admits a non-constant rational solution given in step (9).
    In particular, the output $y_i(x+c)$ in step (9) is $\PS_{3,i}^*$--admissible.
\end{itemize}
Let us now show termination. 
Since $\Sigma^{(1)}:=\Sigma^{(0)} \setminus \Omega_{\mathrm{FG}}^{(0)}$ is a Zariski-closed set, it can be represented by the finite intersection of prime ideals $\mathcal{I}_{\ell}^{(0)}$. 
The canonical representation of $\Pa$ over the new base field $\mathcal{Q}(\mathcal{R}^{(1)})$, where $\mathcal{R}^{(1)}:=\mathcal{R}^{(0)}/\mathcal{I}_{\ell}^{(0)}$, is then used to check~\eqref{eq:FengGaoCases} over $\mathcal{Q}(\mathcal{R}^{(1)})$ in $\V(\mathcal{I}_{\ell}^{(0)})$. Note that in this step, factorization of $A/B$ and equality to zero change (for details on the computation over these fields we refer to Section 10.2 and Appendix B in~\cite{wang2},and~\cite{wang1}). 
It might be the case that the complement $\Sigma^{(1)} \setminus \Omega_{\mathrm{FG}}^{(1)}$ is again non-empty and leads to a further iteration. 
The number of iterative steps is finite because the chain of proper base fields $\mathcal{R}^{(j)}$ is bounded by the number of proper prime ideals $\mathcal{I}_{\ell}^{(j)}$, which is at most $\#|\textbf{a}|$.
Moreover, the prime decomposition in step (7) is finite such that the loop defined by the steps (4)-(8) is finite.
Thus, the termination of Algorithm \textsc{ParamDecomposition} then leads to termination of the algorithm.
\end{proof}

Let us note that $$\PS_1=\PS_1^*, \PS_2 \subseteq \PS_2^*, \PS_3 \supseteq \PS_3^*.$$
 Moreover, $\PS_3^*$ can consist of less or more components than $\PS_3$, i.e., $|I^*| \le |I|$ or $|I^*| \ge |I|$ is both possible. 

\begin{corollary}\label{cor-behaviorSameComponent}
Let $F \in \LL[y,y']$ be as in~\eqref{eq-F} and let 
\begin{equation*}
\PS= {\PS}_1^* \,\dot\cup\, {\PS}_2^*\, \dot\cup\, {\PS}_3^*, \,\,\text{with}\,\,   {\PS}_{3}^*=\dot{\bigcup}_{i\in I^*} {\PS}_{3,i}^*,
\end{equation*} 
be a decomposition w.r.t. rational solutions.  Then for every $\PS_{3,i}^*$ the specialization of the corresponding rational solution $y_i(x)$ at every $\textbf{a}^0 \in \PS_{3,i}^*$ is a rational solution of $F(\textbf{a}^0;y,y')=0$.
\end{corollary}
\begin{proof}
Let $(\PS_{3,i},\Pa_i)$ be the component of a decomposition w.r.t. parametrizations providing $(\PS_{3,i}^*, y_i(x))$. Then, by Proposition~\ref{prop:solSpec}, every specialization of $y_i(x)$ leads to a solution. 
If $\Pa_i$ does not fulfill~\eqref{eq:FengGaoCases}, then components where for every $\textbf{a}^0 \in \PS_{3,i}$ the specialization $\Pa_i(\textbf{a}^0;t)$ does not fulfill~\eqref{eq:FengGaoCases} are in $\PS_2^*$, see Theorem~\ref{thm:FengGaospec}.
\end{proof}

The following theorem shows that the output of Algorithm~\ref{alg-ConstParameters} covers all possible non-constant rational solutions of $F(y,y')=0$ and $F(\textbf{a}^0;y,y')=0$, as long as $F(\textbf{a}^0;y,y')$ remains irreducible.

\begin{theorem}\label{thm:solCovering}
Let $F \in \LL[y,y']$ be as in~\eqref{eq-F}, let $\tilde{F}(y,y')=F(\textbf{a}^0;y,y')=0$ be well--defined and irreducible with a non-constant rational solution $\tilde{y}(x)$. 
Let 
\begin{equation*}
\PS= {\PS}_1^* \,\dot\cup\,  {\PS}_2^*\, \dot\cup\,  {\PS}_3^*, \,\,\text{with}\,\,   {\PS}_{3}^*=\dot{\bigcup}_{i\in I^*}  {\PS}_{3,i}^*,
\end{equation*} 
be a decomposition w.r.t. rational solutions. 
Then there exists $i \in I^*$ such that $\textbf{a}^0 \in \PS_{3,i}^*$ and it holds that $\tilde{y}(x)=y(\textbf{a}^0;x+c)$ for some $c \in \overline{\K}$, where $y(x)$ is the corresponding admissible solution associated to $\PS_{3,i}^{*}$.
\end{theorem}
\begin{proof}
Since $\tilde{F}$ is well--defined and $\tilde{Y}=(\tilde{y}(x),\tilde{y}'(x))$ is a rational parametrization of $\Cu(\tilde{F})$, by Proposition~\ref{cor:componentpara}, there exists $i \in I$ such that, for the corresponding proper parametrization $\Pa$, it holds that $\tilde{Y}(t) = \Pa(\textbf{a}^0;s)$ for some $s \in \overline{\K(\textbf{a}^0)}(t)$. 
{Thus, $\Pa(\textbf{a}^0;t)$ fulfills~\eqref{eq:FengGaoCases}. 
By Theorem~\ref{thm:FengGaospec} (see Remark~\ref{rem-onFengGaoSpec}), also $\Pa$ fulfills~\eqref{eq:FengGaoCases}. 
Let $y(x)$ be the corresponding rational solution. 
From Theorems~\ref{thm:FengGao} and~\ref{thm:FengGaospec} we know that $y(\textbf{a}^0;x)$ defines all rational solutions of $\tilde{F}(y,y')=0$ and there is $c \in \overline{\LL}$ such that $y(\textbf{a}^0,c;x)=y(\textbf{a}^0;x+c)=\tilde{y}(x)$.
}
\end{proof}

Let us illustrate Algorithm~\ref{alg-ConstParameters} by several examples.

{
\begin{example}\label{ex-sevreral steps}
Let us consider 
\[ \begin{array}{ccl} F&= &y'^{3}+\left(-a_{1}^{2}-2 a_{2}+3 a_{3}\right) y'^{2}+\left(-2 a_{1}^{2} a_{3}+2 y a_{1}+a_{2}^{2}-4 a_{3} a_{2}+3 a_{3}^{2}\right) y' \\
\noalign{\vspace*{2mm}} 
&& -a_{1}^{2} a_{3}^{2}+2 y a_{1} a_{3}+a_{2}^{2} a_{3}-2 a_{2} a_{3}^{2}+a_{3}^{3}-y^{2}=0. 
\end{array}\]
In Step (1) of Algorithm~\ref{alg-ConstParameters}, 
\[ \PS=\C^3, \, \mathcal{S}_1=\mathcal{S}_2=\emptyset, \mathcal{S}_3=\C^3\]
and
\[ \Pa=\left(t^{3}+a_{1} t^{2}-a_{2} t, t^{2}-a_{3}\right)\]
is the associated proper parametrization to $\mathcal{S}_3$. In Step (2) we set $\mathcal{S}_{1}^{*}=\mathcal{S}_{2}^{*}=\mathcal{S}_{3}^{*}=\emptyset,$ $\mathcal{R}=\C(a_1,a_2,a_3)$ and $\Sigma=\C^3$.  
In the notation of Theorem~\ref{thm:FengGaospec},
\[ \dfrac{A}{B}:=\dfrac{P_2}{P_{1}^{\prime}}=\dfrac{t^{2}-a_{3}}{3 t^{2}+2 t a_{1}-a_{2}}\]
and hence~\eqref{eq:FengGaoCases} does not hold for $\Pa$ over $\mathcal{R}$. According to Def.~\ref{def-FG} and Theorem~\ref{thm:FengGaospec}, $\Omega_{\mathrm{FG}}$ is the complementary of the algebraic set defined by $\Res_t(A,B)$ over $\C$, that is 
\[ V_1:=\mathbb{V}_{\C}(-4 a_{1}^{2} a_{3}+a_{2}^{2}-6 a_{2} a_{3}+9 a_{3}^{2}). \]
Furthermore, $\Omega_{\mathrm{FG}} \setminus (\PS_1^* \cup \PS_2^* \cup \PS_3^*)=\Omega_{\mathrm{FG}}$ and we replace $\mathcal{S}_{2}^{*}$ by $\C^3\setminus V_1.$

Since $V_1$ is irreducible, in Step (7) we consider the prime ideal $\mathcal{I}=(-4 a_{1}^{2} a_{3}+a_{2}^{2}-6 a_{2} a_{3}+9 a_{3}^{2})$. In Step (8), we replace $\mathcal{R}$ by the quotient field $\C(a_1,a_2,a_3)/\mathcal{I}$ and $\Sigma$ by the surface $V_1$. Let us denote by $\tilde{A}, \tilde{B}$ the polynomials $A,B$ as elements in
$\left(\C(a_1,a_2,a_3)/\mathcal{I}\right)[t]$. 
Then
\[ \dfrac{\tilde{A}}{\tilde{B}}=\dfrac{2 t a_{1}+a_{2}-3 a_{3}}{6 t a_{1}+4 a_{1}^{2}+3 a_{2}-9 a_{3}}\]
Thus,~\eqref{eq:FengGaoCases} does not hold for $\Pa$ over $\C(a_1,a_2,a_3)/\mathcal{I}$. Since
\[ \Res_t(\tilde{A},\tilde{B})=8a_1^3,\]
the new $\Omega_{\mathrm{FG}}$ is $V_1\setminus \mathbb{V}(a_1)$. So, we add $V_1\setminus \mathbb{V}(a_1)$ to $\mathcal{S}_{2}^{*}$. Then, since $V_1\setminus \mathbb{V}(a_1)$ is the complementary in $V_1$ of the line $\mathbb{V}(a_1, -a_2 + 3a_3)$, the new ideal is the irreducible ideal $\hat{\mathcal{I}}=(a_1, -a_2 + 3a_3)$, the new working field is $\C(a_1,a_2,a_3)/\hat{\mathcal{I}}$, and the new subset of the parameter space is $\Sigma=\mathbb{V}(a_1, -a_2 + 3a_3)$. In this situation, let us denote by $\hat{A}, \hat{B}$ the polynomials $\tilde{A},\tilde{B}$ as elements in $\left(\C(a_1,a_2,a_3)/\hat{\mathcal{I}}\right)[t]$. 
Then
\[ \dfrac{\hat{A}}{\hat{B}}=\dfrac{1}{3}.\]
Thus,~\eqref{eq:FengGaoCases} holds for $\Pa$ over $\C(a_1,a_2,a_3)/\hat{\mathcal{I}}$. So, finally we get the rational general solution
\[ y(x)=\left(\dfrac{x}{3}+\dfrac{c}{3}\right)^{3}-3 \left(\dfrac{x}{3}+\dfrac{c}{3}\right) a_{3}\]
of the differential equation
\[ y'^{3}-3 y'^2 a_{3}+4 a_{3}^{3}-y^{2}=0.\]
The decomposition w.r.t. rational solutions is 
\[ \mathcal{S}_1^*=\emptyset, \mathcal{S}_2^*=(\C^3\setminus V_1) \cup \left(V_1\setminus \mathbb{V}(a_1)\right),  \]
and 
\[ \mathcal{S}_3^*=\left\{\left(\mathbb{V}(a_1, -a_2 + 3a_3), y(x)= \left(\dfrac{x}{3}+\dfrac{c}{3}\right)^{3}-3 \left(\dfrac{x}{3}+\dfrac{c}{3}\right) a_{3}\right)\right\} .\]
\end{example}}

\begin{example}\label{ex-1}
Let us consider $$F= 4a_1a_2^2y^4-4a_1a_2y^2y'+a_1y'^2 + a_2y^2 -y'=0$$
and the parameter space $\PS=\C^2$. 
A decomposition w.r.t. parametrization yields $\PS_1=\emptyset, \PS_2= \{(a_1,0) \mid a_1 \in \C, a_1 \ne 0 \}$ and $\PS_3=\PS_{3,1} \,\dot\cup\,\PS_{3,2}\,\dot\cup \,\PS_{3,3}$ where
\[ \begin{array}{l}
\PS_{3,1}  = \left\{\left( \C^2 \setminus \{ (a_1,a_2) \mid a_1a_2=0 \}, 
 \Pa_1 := \left(\dfrac{a_1a_2t}{a_1^3t^2-a_2^3}, \dfrac{(a_1^3t^2 + a_2^3)a_1^2t^2}{(a_1^3t^2 - a_2^3)^2}\right) \right)\right\} \\
\noalign{\vspace*{3mm}}
\PS_{3,2} =\left\{ \left( \{(0,a_2) \mid a_2 \in \C, a_2 \ne 0 \}, \Pa_2:=(t,a_2t^2)\right)\right\}
\\
\noalign{\vspace*{3mm}}
\PS_{3,3}=\left\{\left( \{(0,0)\}, \Pa_3:=(t,0)\right)\right\}.
\end{array}
\]
Note that $\Pa_1(0,a_2;t)=(0,0)$ and $F(0,a_2;y,y')=a_2 y^2-y'$.
In the case of a specialization $\ {\textbf{a}^0}:=(a_1,0) \in \PS_2$, the curve factors into lines since $F({\textbf{a}^0};y,y')=y'(a_1y'-1)$.

Using~\eqref{eq:FengGaoCases} for $\Pa_{1}=(P_1,P_2)$, we see that $$P_2/P_1' = \dfrac{-a_1t^2}{a_2}$$ leads to the rational solution $$y_1(x)=\frac{x+c}{a_1-a_2(x+c)^2}$$  and $\Omega_{\mathrm{FG}}=\PS_{3,1}$ is added to $\PS_3^*$.

For $\Pa_{2}$ we obtain that~\eqref{eq:FengGaoCases} is fulfilled with $\alpha=a_2, \beta=0$ leading to the rational solution $$y_2(x)=-\dfrac{1}{a_2(x+c)}.$$  Note that for $\textbf{a}^0=(0,a_2)$, with $a_2 \ne 0$, the solutions $y_1(x), y_2(x)$ coincide.
For ${\textbf{a}^0}=(0,0) \in \PS_{3,3}$, the specialization $F({\textbf{a}^0};y,y')=-y'$ defines a vertical line. 
Verifying~\eqref{eq:FengGaoCases} for $\Pa_2$ leads to $\alpha=0$ and the solutions are the constants (cf. Remark~\ref{rem-zero}).


We thus obtain the decomposition w.r.t. rational solutions
\[ \PS_1^*=\emptyset, \PS_2^* = \PS_2 \cup \{(0,0)\}, \PS_3^* = \PS_{3,1}^* \cup \PS_{3,2}^* \]
where 
\[ \PS_{3,1}^*=\{(\PS_{3,1},y_1)\}, \,\,  \PS_{3,2}^*=\{(\PS_{3,2},y_2)\}. \]
Let us note that in the case of ${\textbf{a}^0} \in \PS_2$, the specialization $\Pa_1({\textbf{a}^0};t)=(0,1/a_1)$ is not a rational parametrization of a component of $F({\textbf{a}^0};y,y')=y'(a_1 y'-1)$ anymore. 
For $\Pa_2({\textbf{a}^0};t)=(t/a_1,1/a_1)$, however, we find the zero $$y({\textbf{a}^0};x)=\dfrac{x+c}{a_1}$$ of $F({\textbf{a}^0};y,y')$ (cf. Proposition~\ref{prop:solSpec}).
\end{example}

\begin{example}\label{ex-localsols}
Let $F=2y-y'^2+2a_1y'-a_2^2+a_2(2a_1-2y')$. 
Then the decomposition w.r.t. parametrization is $\PS_3=\C^2$ and consists of only one component with corresponding parametrization 
$$\Pa = \left(\dfrac{t^2}{2}-a_1t,t-a_2\right).$$ 
Equation~\eqref{eq:FengGaoCases} is generically not fulfilled because $$\dfrac{A}{B}:=\dfrac{P_2}{P_1'} = \dfrac{t-a_2}{t-a_1}.$$
The leading coefficients are one and $\Res_t(t-a_2,t-a_1)=a_2-a_1$,  and we obtain $\Omega_{\mathrm{FG}}=\Omega_{\nonZ(\Res(A,B))}=\PS \setminus \{ (a_1,a_2) \in \C^2 \mid a_1=a_2 \}$. 
{On the line $$\{ (a_1,a_2) \in \C^2 \mid a_1=a_2 \}$$
the associated polynomial ideal generated by $a_1-a_2$ is prime. Let us now consider $\Pa$ over $\C[a_1,a_2]/(a_1-a_2)$. Then, it holds that $\dfrac{A}{B}=1$. 
In this case,~\eqref{eq:FengGaoCases} is fulfilled with $\alpha=1$} leading to the solution $$\tilde{y}(x)=\dfrac{x^2}{2}-ax$$ of $F((a,a);y,y')= 2y-y'^2+a^2=0$.
Thus, a decomposition w.r.t. rational solutions is $\PS_1^*=\emptyset, \PS_2^* = \C^2 \setminus \{ (a,a) \mid a \in \C \}$, $\PS_3^* = \{ (a,a) \mid a \in \C \}$ with corresponding rational solution $\tilde{y}(x)$.
\end{example}

\section*{Acknowledgements}
Authors thank the anonymous referee for his/her comments.

Authors partially supported by the grant PID2020-113192GB-I00/AEI/\-10.13039/\-501100011033 (Mathematical Visualization: Foundations, Algorithms and Applications) from the Spanish State Research Agency (Ministerio de Ciencia, Innovación y Universidades). 
First author also supported by the OeAD project FR 09/2022.

\bibliographystyle{acm}

\appendix

\section{Decomposition w.r.t. surjectivity}\label{sec-decomposition-surjectivity}

In Subsection \ref{sec-decomposition} we have seen how to decompose the parameter space  in terms of rational parametrizations. With this decomposition we have approached the problem of finding rational solutions of the differential equation. As a future line of working one may investigate the possibility of proceeding similarly for local solutions. As a particular important case, one still may consider rational parametrizations but, in that case, one needs to guarantee that all points on the curve are covered by the parametrizations. Motivated by this fact, in this appendix
we study how to refine the decomposition of $\PS_3$, in~\eqref{eq-decomp2}, Def. \ref{def-dec-param}, in order to guarantee  that the parametrizations provide a surjective covering.

Let $\Pa(t)$ be a proper rational parametrization of and algebraic curve $\Cu(F)$ over $\LL$. 
By~\cite[Theorem 6.22]{SWP08}, it holds that $\Cu(F) \setminus \Pa(\overline{\LL})$ contains at most one point. 
If this point exists, we call it the \textit{critical point} of $\Pa(t)$. 
Furthermore, by~\cite[Theorem 6.26]{SWP08}, one can always {obtain} a surjective parametrization. However, the field of parametrization of this parametrization is, in general, an algebraic field extension of the field of parametrization of $\Pa$. So, in general, one may have to introduce an algebraic element depending on $\mathbf{a}$. To avoid this, when $\Pa(t)$ is not surjective, we may work with a finite collection of parametrizations, such that the union of their images cover the curve. 
The following lemma ensures that the coefficients of these parametrizations can still be assumed to be in the field of parametrization.

\begin{lemma}\label{lem:covering}
Let $\LL$ be a field, $F(y,z) \in \LL[y,z]$, $\Cu(F)$ be a rational curve and let $\FF$ be a parametrizing field of $\Cu(F)$. 
Then, there exists a set of proper parametrizations $\{ \Pa_i(t) \}_{i \in I} \subset \FF(t)^2$, with $\#(I) \le 2$, such that $\Cu(F)= \bigcup_{i \in I} \Pa_i(\overline{\LL})$.
\end{lemma}
\begin{proof}
We assume w.l.o.g. that $\Cu(F)$ is neither a vertical nor a horizontal line. So, in the following, none component of the parametrizations is constant.
Since $\FF$ is a parametrizing field, let $\Pa_1(t)\in \FF(t)^2$ be a proper parametrization. 
Let
\[ \Pa_1=\left(\dfrac{p_1}{q_1},\dfrac{p_2}{q_2} \right), \]
with $\gcd(p_i,q_i)=1$.
If there exists $i\in \{1,2\}$ such that $\deg(p_i)>\deg(q_i)$ then $\Pa_1(\overline{\LL})=\Cu(F)$ (see~\cite[Corollary 6.20]{SWP08}), and the statement follows with $I=\{1\}$.

Let us assume that $\deg(p_i)\leq \deg(q_i)$ for $i\in \{1,2\}$.  
Moreover, let us also assume that none of the polynomials $p_1,p_2,q_1,q_2$ has zero as a root; if this would be the case, we can apply a change $\Pa_1(t+a)$ with $a\in \LL$.

Let us express the polynomials $p_i,q_j$ as
\[ p_1=\sum_{i=0}^{r}a_i t^i,\, q_1=\sum_{i=0}^{n}b_i t^i,\, p_2=\sum_{i=0}^{s}c_i t^i,\, q_2=\sum_{i=0}^{m}d_i t^i, \]
where $a_ra_0b_nb_0c_sc_0d_md_0\neq 0$. 
Then, by~\cite[Theorem 6.22]{SWP08},
\[\Cu(F)\setminus \{ (a_r/b_n,c_s/d_m)\} \subset \Pa_1(\overline{\LL}).\]
Now, let $\mu\in \FF$ be such that $p_1(\mu)b_n-q_1(\mu)a_r\neq 0$ and $p_1(\mu)q_1(\mu)p_2(\mu)q_2(\mu)\neq 0$; this is possible because $b_n,a_r,p_1,q_1,p_2,q_2$ are not zero. 
We consider the parametrization $\Pa_2(t)=\Pa_1(1/t+\mu)$. That is
\[\begin{array}{ccr} \Pa_2(t) &=&\left(\dfrac{(a_r+\tilde{a}_{r-1} t+ \cdots+\tilde{a}_{1} t^{r-1} +p_1(\mu)t^r)t^{n-r}}{b_n+\tilde{b}_{n-1} t+ \cdots+\tilde{b}_{1} t^{n-1} +q_1(\mu)t^n}, \right.\\
\noalign{\vspace*{2mm}}
&& \left. \dfrac{(c_s+\tilde{c}_{s-1} t+ \cdots+\tilde{c}_{1} t^{s-1} +p_1(\mu)t^r)t^{n-r}}{d_m+\tilde{d}_{m-1} t+ \cdots+\tilde{d}_{1} t^{m-1} +q_1(\mu)t^n}\right), \end{array} \]
for some $\tilde{a}_i, \tilde{b}_i, \tilde{c}_i, \tilde{d}_i\in \FF$.
Now, $\Cu(F)\setminus \{ (p_1(\mu)/q_1(\mu), p_2(\mu)/q_2(\mu))\} \subset \Pa_2(\overline{\LL})$. 
Since $p_1(\mu)/q_1(\mu)\neq a_r/b_n$, the statement follows for $I=\{1,2\}$.
\end{proof}

\begin{remark}
Throughout this paper, when we speak about a surjective rational covering, we will mean the covering provided by the proof of Lemma~\ref{lem:covering}.
\end{remark}

For a given proper rational parametrization $\Pa$ of $\Cu(F)$, we consider the open set $\Omega_{\proper(\Pa)} \subseteq \PS$, introduced before (see Subsection \ref{sec-pre}). 
In addition, let $\{ \Pa_i(t) \}_{i \in I} \subset \FF(t)^2$, with $\#(I) \le 2$  be a surjective rational covering of a given curve $\Cu(F)$.  Then, we introduce a new open subset in the following way. Let $C_i:=(A_{i,1},A_{i,2})\in \FF^2$ be the  critical point of $  \Pa_i(t)$; and let $\Pa_i$ be expressed in reduced form as $(p_{i,1}/q_{i,1},p_{i,2}/q_{i,2})$. We consider the polynomials
\[ g_{i,j}:=A_{i,j} q_{i,j}-p_{i,j}, \,\,\,\text{$i\in I$ and $j\in \{1,2\}$}, \]
and 
\[ g_i:=\gcd(g_{i,1},g_{i,2}),\, R_i:=\Res_t(g_i, q_{i,1} q_{i,2})), \,\,\, i\in \{1,2\}. \] 
Then, we define the open subset of $\PS$
\[ \Omega_{\surj}(\{ \Pa_i(t) \}_{i \in I}):=\bigcap_{i\in I} \left(\Omega_{\gcd(g_{i,1},g_{i,2})} \cap \Omega_{\nonZ(R_i)}
\right).
\]

\begin{lemma}\label{lem-specializationCovering}
Let $\{ \Pa_i(t) \}_{i \in I} \subset \FF(t)^2$, with $\#(I) \le 2$, be a surjective rational covering of a given curve $\Cu(F)$ such that $\Cu(F)= \bigcup_{i \in I} \Pa_i(\overline{\LL})$. 
Let $\textbf{a}^0 \in \bigcap_{i \in I} \Omega_{\proper(\Pa_i)} \cap \Omega_{\surj}(\{ \Pa_i(t) \}_{i \in I}) \subset \PS$.
Then $$\Cu(F,\textbf{a}^0)= \bigcup_{i \in I} \Pa_i(\overline{\K}).$$
\end{lemma}
\begin{proof}
	Since $\textbf{a}^0\in \Omega_{\proper(\Pa_i)}$, by~\cite[Theorem 5.5]{falkensteiner2023rationality}), $\Pa_i(\mathbf{a}^0;t)$ parametrizes properly $\Cu(\mathbf{a}^0;F)$. Moreover, the numerators and denominators of $\Pa_i(t)$ stay coprime (see proof of Theorem 5.5 in~\cite{falkensteiner2023rationality}). 
	Furthermore, the degrees of the numerators and denominators are also preserved. So, the critical point of $\Pa_i(t,\mathbf{a}^0)$ is the specialization of the critical point of $\Pa_i(t)$, namely $C_i(\mathbf{a}^0)$. 
	It remains to prove that $C_i(\mathbf{a}^0)$ is reachable by $\Pa_j(t,\mathbf{a}^0)$ for some $j\in I$. By hypothesis, there exists $t_0\in \overline{\LL}$, and $j\in I$ such that $\Pa_j(t_0)=C_i$. In particular, $g_j(t_0)=0$ and hence $\deg_t(g_j)>0$. On the other hand, since $\mathbf{a}^0\in \Omega_{\gcd(g_{j,1},g_{j,2})}$, by \cite[Corollary 3.8]{falkensteiner2023rationality}), $g_j(t,\mathbf{a}^0)=\gcd(g_{j,1}(t,\mathbf{a}^0),g_{j,2}(t,\mathbf{a}^0))$ and $\deg(\gcd(g_{j,1}(t,\mathbf{a}^0),g_{j,2}(\mathbf{a}^0;t)))=\deg(g_j(\mathbf{a};t))>1$. 
	Let $t_1\in \overline{\K}$ be a root of $\gcd(g_{j,1}(\mathbf{a}^0;t),g_{j,2}(\mathbf{a}^0;t))$. Since $\mathbf{a}^0\in \Omega_{\nonZ(R_j)}$, it holds that $q_{j,1}(\mathbf{a}^0;t_1)q_{j,2}(\mathbf{a}^0;t_1)\neq 0$. Thus, $\Pa_j(\mathbf{a}^0;t_1)=C_i(\mathbf{a}^0)$.
\end{proof}

Using the previous result, we can further decompose the parameter space under the surjectivity  criterion.

{\begin{definition}\label{def-dec-surj}
Let $F\in \LL[y,z]$ be irreducible. 
		Let $I\subset \N$ be finite. For $i\in I$,  let  $\tilde{\PS}_1,\tilde{\PS}_2,\tilde{\PS}_{3,i}\subset \PS$ be disjoint sets,  and let 
		\[  \tilde{\PS}_{3}=\displaystyle{\dot{\bigcup}_{i\in I} \{(\tilde{\PS}_{3,i},\{\tilde{\Pa}_{i,j}(\textbf{a},t) \}_{j\in J_i} )\}} \] 
		where $\{\tilde{\Pa}_{i,j} \}_{j\in J_i}$ is a finite set of $\tilde{\PS}_{3,i}$--admissible parametrizations. 
		We say that 
		\begin{equation}\label{eq-decomp3}
			\tilde{\PS}_1 \,\dot\cup\, \tilde{\PS}_2\, \dot\cup\, \tilde{\PS}_3, 
		\end{equation} 
		is a \textit{decomposition w.r.t. surjective (rational) parametrizations of $\Cu(F)$} if  
		\begin{itemize}
			\item[(a)] $\displaystyle{\PS=\tilde{\PS}_1 \,\dot\cup\,\tilde{ \PS}_2\,\dot\cup\, \tilde{\PS}_3.}$
			\item[(b)]
			For every specialization $\textbf{a}^0\in \tilde{\PS}_i$, case (i) holds:
			\begin{enumerate}
				\item either $F(\textbf{a}^0;y,z)$ is not well--defined or $F(\textbf{a}^0;y,z)\in \overline{\K}$;
				\item the genus of  $\Cu(\textbf{a}^0;F)$ is positive, or $F(\textbf{a}^0;y,z)$ is reducible (over $\overline{\K}$);
				\item the genus of $\Cu(\textbf{a}^0;F)$ is zero and $\{\tilde{\Pa}_{i,j}\}_{j \in J_i}$ is a surjective (proper) rational covering of $\Cu(\textbf{a}^0;F)$; that is $\tilde{\Pa}_{i,j}$ is a $\tilde{\PS}_i$--admissible proper parametrization, for all $j \in J_i$, and 
				\[\Cu(\textbf{a}^0;F)=\bigcup_{j\in J_i} \tilde{\Pa}_{i,j} (\overline{\K}). \]
			\end{enumerate}
		\end{itemize}
\end{definition}}

	\begin{remark}\label{rem-eq-decomp3} \,
		\begin{enumerate}
			\item The main difference between a decomposition w.r.t. parametrizations (see Definition~\ref{def-dec-param}) and w.r.t. surjective parametrizations (see Definition~\ref{def-dec-surj}) is that, in the first case, each $\PS_{3,i}$ contains a rational parametrization that specializes properly while, in the second case, each $\tilde{\PS}_{3,i}$ contains a finite set of parametrizations which union of images covers the whole curve and the property is preserved under specializations.
			\item For computing a decomposition w.r.t. surjective (rational) pa\-ra\-me\-tri\-za\-tions, one may proceed as follows. We consider a decomposition w.r.t. (rational) parametrizations (see Definition \ref{eq-decomp2} and Remark \ref{rem-def-dec-param}). Then, we take $\tilde{\PS}_1=\PS_1$ and $\tilde{\PS}_2=\PS_2$. Now, in $\PS_3$, for each parametrization $\Pa_i$, associated to the component $\PS_{3,i}$, we apply Lemma~\ref{lem:covering} to get $\{\tilde{\Pa}_{i,j} \}_{j\in J_i}$ and we replace, in the construction in~\cite[Section 6]{falkensteiner2023rationality}, $\Omega_{\proper(\Pa_i)}$ by $$\Omega_i:=\bigcap_{j\in J_i} \Omega_{\proper(\tilde{\Pa}_{i,j})} \cap \Omega_{\surj}(\{ \tilde{\Pa}_{i,j}(t) \}_{j \in I_j}).$$
			Then, eventually a finite number of constructible sets $\tilde{\PS}_{3,i}$ is achieved. For this purpose, by an iterative construction, we adjoin to every $\tilde{\PS}_i$ and $\tilde{\PS}_{3,i}$ a computable field $\FF_{\mathrm{J}}$, where $\mathrm{J}$ denotes an ideal represented by a Gr\"obner basis, such that every specialization $\textbf{a}^0 \in \tilde{\PS}_1, \textbf{a}^0 \in \tilde{\PS}_{2}$ or $\textbf{a}^0 \in \tilde{\PS}_{3,i}$, respectively, can be treated simultaneously and leads to an algorithmic treatment (see also~\cite[Section 6]{falkensteiner2023rationality}).
		\end{enumerate}
\end{remark}

\begin{theorem}\label{thm-surjectiveSpecialization}
	Let $F \in \K[y,z]$ be irreducible and let $$\PS=\tilde{\PS}_1 \,\dot\cup\, \tilde{\PS}_2\, \dot\cup\, \tilde{\PS}_3,\, \,\,\text{with}\,\,\,\tilde{\PS}_{3}=\dot{\bigcup}_{i\in I} \tilde{\PS}_{3,i},$$ be a decomposition w.r.t. surjective parametrizations of $\Cu(F)$.	Then, there exists a set $\{ \Pa_{i,j}(t) \}_{i \in I, j \in J_i}$, with $\#(J_i) \le 2$, such that for every $\textbf{a}^0 \in \tilde{\PS}_{3,i}$
	\begin{enumerate}
		\item $\Pa_{i,j}(\textbf{a}^0;t)$, $j \in J_i$, are proper parametrizations of $\Cu(\textbf{a}^0;F)$; and
		\item $\Cu(\textbf{a}^0;F)= \bigcup_{j \in J} \Pa_{i,j}(\overline{\K(\textbf{a}^0)})$.
	\end{enumerate}
\end{theorem}
\begin{proof}
	From~\cite[Remark 6.1]{falkensteiner2023rationality} we obtain the decomposition of the parameter space $\PS= \PS_1 \, \dot\cup \,\PS_2\, \dot\cup\, \PS_3$ and proper parametrizations $\Pa_i$ such that $\Pa_i(\textbf{a}^0;t)$ is a proper parametrization of $\Cu(\textbf{a}^0;F)$ for every $\textbf{a}^0 \in \Omega_{\proper(\Pa_i)}$. 
	For every $\Pa_i$, by Lemma~\ref{lem:covering}, there are $\Pa_{i,j}$, $j \in J_i$ with $\#(J_i) \le 2$, such that $\Cu(\textbf{a}^0;F)= \bigcup_{j \in J}\Pa_{i,j}(\overline{\K})$. 
	Then the result follows from Lemma~\ref{lem-specializationCovering}.
\end{proof}

In this situation, we can execute Algorithm \ref{alg-ConstParameters}, applying to the output of Step 1 the computational approach described  in Remark \ref{rem-eq-decomp3}, so that we get a decomposition w.r.t. rational solutions  through a decomposition w.r.t. surjective (rational) parametrizations of $\Cu(F)$. 
 

\end{document}